\documentclass{ifacconf}

\usepackage{amsmath,amsfonts,amssymb,physics}
\usepackage[caption=false,font=normalsize,labelfont=sf,textfont=sf]{subfig}
\newcommand{\x}{\mathsf{x}}
\newcommand{\y}{\mathsf{y}}

\newcommand{\R}{\mathbb{R}}
\newcommand{\C}{\mathcal{C}}

\newcommand{\dotprod}[2]{\langle #1,#2\rangle}
\DeclareMathOperator{\atantwo}{atan2}

\newtheorem{definition}{Definition}{}
{}
{}
{}
{}
{}
\newtheorem{theorem}{Theorem}{}
\newenvironment{proof}{%
  \par\noindent\textbf{Proof. }%
}{%
  \hfill$\square$\par
}

\DeclareMathOperator*{\argmin}{arg\,min}

\usepackage{graphicx}      
\usepackage{natbib}        
\begin{document}
\begin{frontmatter}

\title{Collision-Aware Density-Driven Control of Multi-Agent Systems via Control Barrier Functions} 

\thanks[footnoteinfo]{This work was supported by NSF CAREER Grant CMMI-DCSD-\#2145810. This paper has been accepted for presentation at the 5th Modeling, Estimation and Control Conference (MECC 2025), Pittsburgh, USA.}

\author[First]{Sungjun Seo} 
\author[Second]{Kooktae Lee}

\address[First]{New Mexico Institute of Mining and Technology, 
   Socorro, NM 87801 USA (e-mail: sungjun.seo@student.nmt.edu).}
\address[Second]{New Mexico Institute of Mining and Technology, 
   Socorro, NM 87801 USA (e-mail: kooktae.lee@nmt.edu), Corresponding author.}

\begin{abstract}                This paper tackles the problem of safe and efficient area coverage using a multi-agent system operating in environments with obstacles. Applications such as environmental monitoring and search and rescue require robot swarms to cover large domains under resource constraints, making both coverage efficiency and safety essential. To address the efficiency aspect, we adopt the Density-Driven Control (D$^2$C) framework, which uses optimal transport theory to steer agents according to a reference distribution that encodes spatial coverage priorities. To ensure safety, we incorporate Control Barrier Functions (CBFs) into the framework. While CBFs are commonly used for collision avoidance, we extend their applicability by introducing obstacle-specific formulations for both circular and rectangular shapes. In particular, we analytically derive a unit normal vector based on the agent's position relative to the nearest face of a rectangular obstacle, improving safety enforcement in environments with non-smooth boundaries. Additionally, a velocity-dependent term is incorporated into the CBF to enhance collision avoidance. Simulation results validate the proposed method by demonstrating smoother navigation near obstacles and more efficient area coverage than the existing method, while still ensuring collision-free operation.
\end{abstract}


\end{frontmatter}

\section{Introduction}

Swarm robotics has gained attention as a framework where large groups of simple agents collaborate to perform complex tasks \citep{beni2005swarm}. The properties of decentralized control, local sensing, and autonomous coordination allow swarm systems to demonstrate flexibility, robustness, and scalability \citep{brambilla2013swarm}. These features are crucial for applications in dynamic environments such as disaster response \citep{arnold2020heterogeneous}, pollution monitoring \citep{aznar2014modelling}, and wildlife surveillance \citep{kabir2021wildlife}, where agents must explore large areas with limited resources. Additionally, swarm systems can provide high levels of fault tolerance, which is especially beneficial in unpredictable or hazardous environments.

A key challenge in these missions is balancing efficiency and safety. Efficiency is critical due to constraints on agent count, energy, and communication range, while safety ensures agents avoid collisions in cluttered environments. Recent research has introduced Density-Driven Control (D$^2$C), a framework based on Optimal Transport (OT) theory, to enable multi-agent systems to distribute themselves according to a reference density that represents regional importance. This ensures that agents prioritize high-density regions and avoid redundant coverage in low-density areas. D$^2$C has been particularly useful for missions that require non-uniform area coverage, where different regions have varying levels of importance.

However, D$^2$C does not inherently handle collision avoidance. While classical methods like Artificial Potential Fields (APFs) are commonly used for this purpose, they often suffer from oscillations and local minima, limiting their reliability in dynamic environments. To address these challenges, Control Barrier Functions (CBFs) provide a more systematic framework for enforcing safety by ensuring forward invariance of a safe set. Most CBF-based collision avoidance studies model obstacles and agents as circular or elliptical shapes to simplify the mathematics, limiting their effectiveness in environments with more complex geometries.


To integrate safety into the D$^2$C framework, this paper proposes a novel enhancement of the conventional CBF approach. The key contributions are as follows. Firstly, generalized CBF formulations are introduced for rectangular obstacles, offering a more realistic representation of environments like buildings and urban landscapes. Secondly, the unit normal vector required for the CBF is derived analytically, based on the agent's relative position to the nearest face of the obstacle. Thirdly, a velocity-dependent term is incorporated into the CBFs to enhance collision avoidance by regulating the agent's speed near obstacles. Lastly, the CBFs are seamlessly integrated with D$^2$C, resulting in a unified control strategy that ensures both efficient coverage and safety, allowing the system to adapt to dynamic environmental conditions.

Simulation results demonstrate that the D$^2$C+CBF method outperforms the D$^2$C+APF baseline, delivering superior safety and coverage efficiency in complex environments. The proposed approach ensures that agents not only achieve effective area coverage but also robustly avoid obstacles, including inter-agent collisions, adapting to dynamic changes in the environment.

These contributions offer a robust solution for safe and efficient multi-agent area coverage in realistic environments with complex obstacles, providing a valuable tool for various real-world applications.

\section{Preliminaries}
\noindent{Notations}:  The sets of real numbers and integers, respectively, are denoted by $\mathbb{R}$ and $\mathbb{I}$. The set of natural numbers is denoted by $\mathbb{N}$. The set of positive real numbers is represented by $\R_{>0}$. The integers in the interval $[n,m]$ are represented by $\mathbb{I}_{n:m}$. The notation $\{q_j\}_{j=1}^{N}$ represents the set of elements $\{q_1,q_2,\cdots,q_N\}$. The inner product of the vectors $\mathbf{a}$ and $\mathbf{b}$ is represented by $\langle \mathbf{a},\mathbf{b}\rangle$, and the norm of $\mathbf{a}$ is given by $\norm{\mathbf{a}}=\sqrt{\langle \mathbf{a},\mathbf{a}\rangle}$. The function $\atantwo(\mathbf{a})$, where $\mathbf{a}$ is a vector $(x, y)$, returns the counterclockwise angle (from $-\pi$ to $\pi$) between the positive $x$-axis and the vector $(x, y)$ in the Cartesian plane.

\subsection{Area Coverage Problem}

The area coverage problem in multi-agent systems involves coordinating agents to efficiently explore and cover a designated area over time.

A common approach is uniform area coverage, where agents systematically cover the domain using methods like the lawnmower technique. While effective in ideal conditions, uniform strategies are impractical in real-world missions with limited agent resources, such as the number of agents available for deployment and their operational time.

In scenarios like search and rescue operations following natural disasters, large and cluttered environments combined with resource constraints make rapid, intelligent responses crucial. 
Non-uniform area coverage addresses these challenges. Given the regional priority map shown in Fig. \ref{fig: concept}(a), non-uniform coverage can be achieved by allocating more agent coverage effort to high-priority regions, as illustrated in Fig. \ref{fig: concept}(b).
By integrating constraints such as agents' available coverage time, communication range, and agent count, the coverage strategy becomes more practical and efficient.

The D$^2$C framework, introduced in \cite{kabir2020receding, kabir2021efficient, kabir2021wildlife, lee2022density}, provides a solution for non-uniform area coverage using multi-agent systems. Built on Optimal Transport (OT) theory, D$^2$C steers agents to match a reference density map, which can be derived from sources like satellite imagery and cellular signals, highlighting regions of higher priority.

\begin{figure}[!h]
    \centering
    \includegraphics[width=0.9\linewidth]{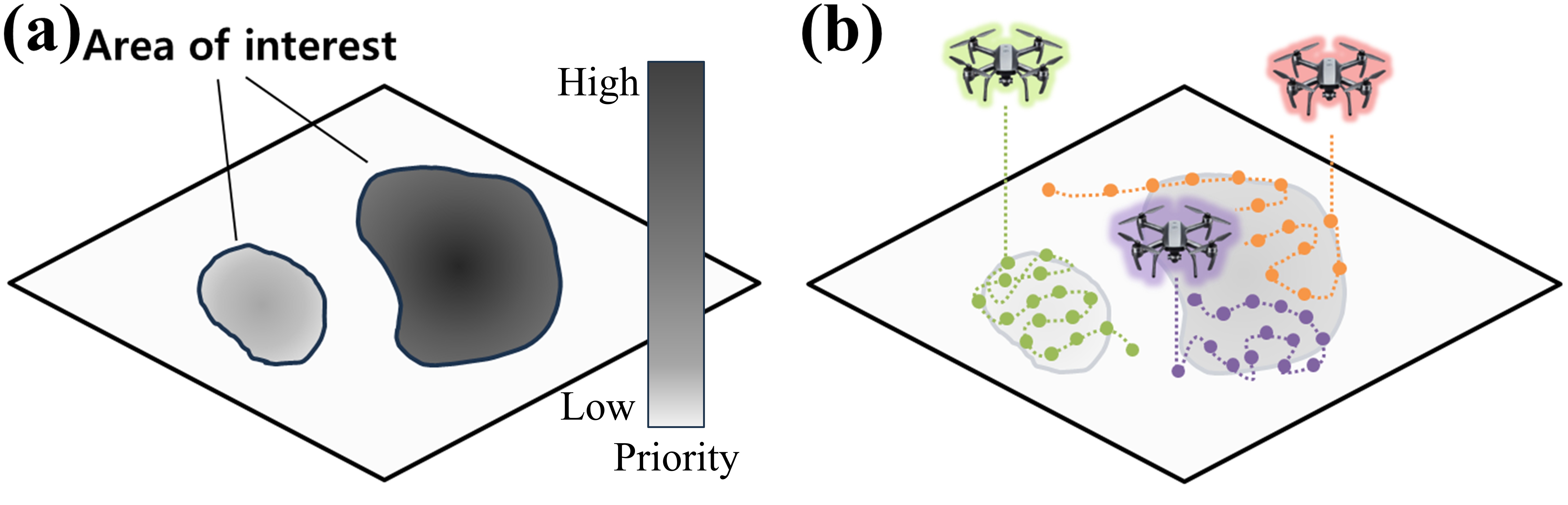}
    \caption{Conceptual illustration of non-uniform area coverage: (a) Pre-determined priority map; (b) Example of non-uniform coverage, where dashed lines represent continuous agent trajectories and the dots indicate discrete-time positions.}
    \label{fig: concept}
\end{figure}


Unfortunately, the original D$^2$C formulation does not account for collision avoidance, a critical concern in dense or constrained environments. To address this, we extend D$^2$C with a collision avoidance mechanism based on CBFs, ensuring both safety and efficiency in area coverage.

To facilitate a better understanding of its application in non-uniform area coverage, a brief overview of D$^2$C is provided in the following section.

\subsection{Density-Driven Control (D$^2$C)}

Optimal Transport (OT) theory is a mathematical framework used to measure the distance between two distributions. The Wasserstein distance, a key concept in OT, can be thought of as finding the most efficient way to move mass from one distribution to another while minimizing the total cost. Alternatively, the Wasserstein distance can be interpreted as a metric that quantifies the distance between two distributions. In multi-agent systems, this distance helps guide agents to align their temporal distribution with a given reference distribution represented by sample points. The D$^2$C framework, developed in \cite{kabir2020receding, kabir2021efficient, kabir2021wildlife, lee2022density}, leverages OT theory to enable agents to align their positions with a reference distribution that encodes spatial priorities for coverage, such as areas requiring more attention or resources.

The D$^2$C method operates iteratively, with each agent adjusting its position based on a local set of reference points, known as sample points (SPs), each associated with specific weights. These weights reflect the importance of the corresponding areas in the environment. While the detailed description of D$^2$C can be found in \cite{kabir2020receding, kabir2021efficient, kabir2021wildlife, lee2022density}, we provide a brief conceptual overview here, given the space limitations.

\begin{figure*}[h]
    \centering
    \includegraphics[width=0.85\linewidth]{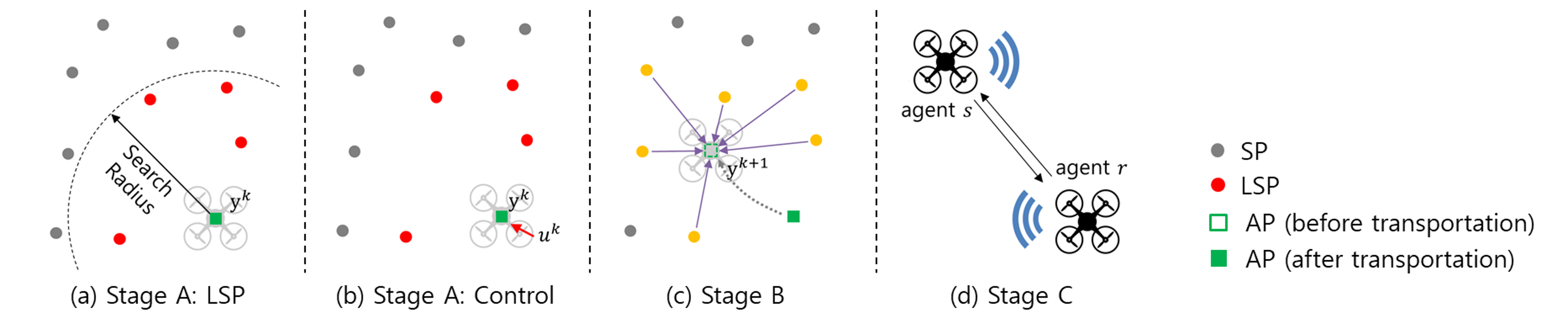}
    \caption{Conceptual illustration of three-stage D$^2$C procedure}
    \label{fig:D2C illustration}
\end{figure*}

\subsubsection{Stage A: Optimal Control Stage}  
In this stage, each agent determines its next goal position by evaluating the cost associated with potential waypoints generated from nearby sample points, referred to as local sample points (LSPs). The LSPs are selected based on the weight-normalized distance (wnD) metric, which considers both the distance to the sample point and the remaining weight of that point. Each agent extends its search radius to include a specific number of LSPs. Once the LSPs are selected, the agent constructs possible waypoints, and the cost of each waypoint is evaluated. The waypoint with the minimum cost is chosen, and the agent sets its goal position as the first sample point in the selected waypoint. The agent then moves towards this goal position using a control law that accounts for its dynamics and the selected goal.

\subsubsection{Stage B: Weight Update Stage}  
After reaching the goal position, each agent updates the weights of the sample points based on its movement. The weights are adjusted by transporting them to the new position while minimizing the Wasserstein distance, ensuring that areas of high priority are covered as efficiently as possible. This update process enables each agent to track which regions still require attention and reinforces alignment with the reference distribution. 

Through this stage, the coverage progress is dynamically updated, allowing the system to iteratively adapt its strategy in response to evolving coverage and focus on under-covered areas in subsequent planning cycles.

\subsubsection{Stage C: Weight-sharing Stage}  
When agents are within communication range of one another, they share their updated weights in a decentralized manner. This localized communication allows each agent to operate independently without requiring access to global information, while still enabling coordinated behavior. By exchanging only local data, agents synchronize their coverage efforts and maintain a shared understanding of which regions have been explored. The weight-sharing process ensures that each agent has access to the most recent information about the environment, improving overall coverage efficiency while preserving the decentralized nature of the system.

The conceptual illustration for the three-stage D$^2$C procedure is shown in Fig. \ref{fig:D2C illustration}.

\section{Control Barrier Functions (CBFs)-based Collision Avoidance}

\subsection{Obstacle Models}\label{sec: Obstacle Models}
In this study, two types of obstacles are considered: circular and rectangular shapes. These shapes define the regions that robots must avoid to ensure collision-free operation. While more complex obstacle shapes may arise in practical applications, they can generally be approximated by these basic forms or their combinations, provided an appropriate safety margin is applied. Fig. \ref{fig: Obstacles} illustrates the dimensional parameters of both circular and rectangular obstacles.
\vspace{-0.1in}

\begin{figure}[h]   
\centering\includegraphics[width=0.8\linewidth]{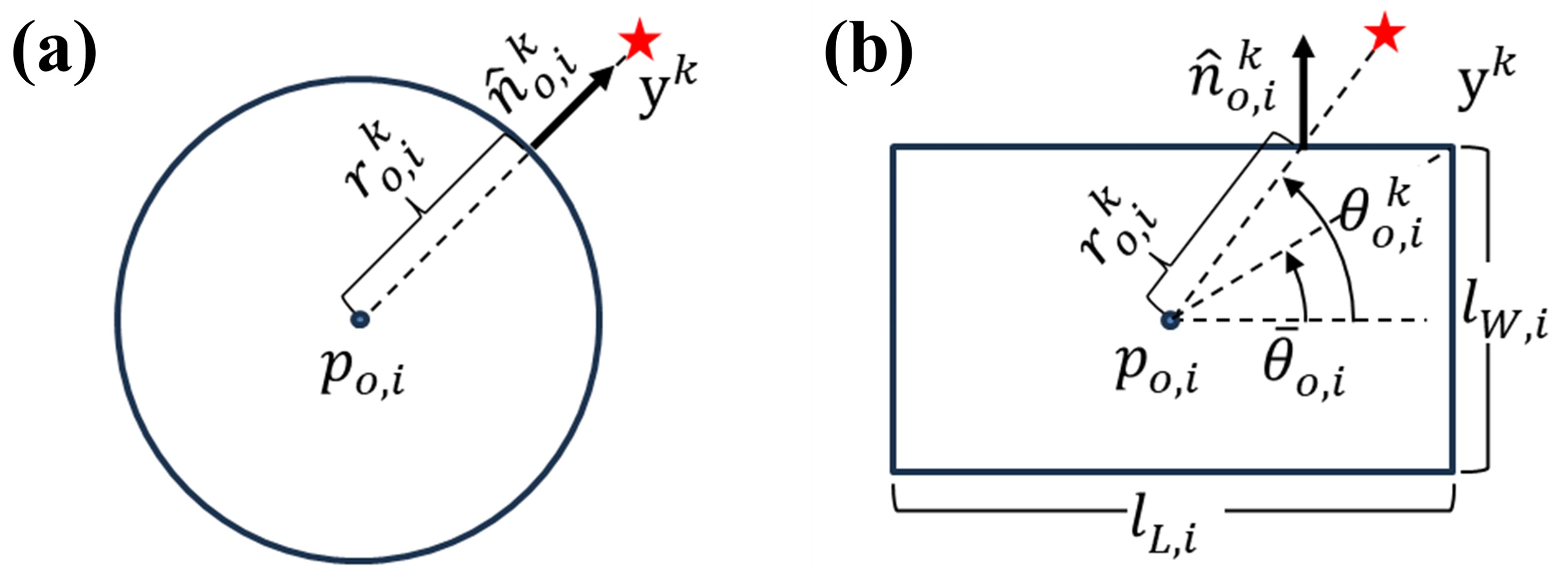}\caption{dimensional parameters of obstacles. (a) Circular obstacle. (b) Rectangular obstacle.}    \label{fig: Obstacles}\end{figure}

In Fig.~\ref{fig: Obstacles}, \( p_{o,i} \) represents the mass center of each obstacle, where \( i = 1, 2, \dots, O \) denotes the obstacle index. Furthermore, the right superscript $k$ is used throughout the paper to indicate the time step. The variable \( r^k_{o,i} \) indicates the distance from \( p_{o,i} \) to the intersection point between the obstacle boundary and the line extending from \( p_{o,i} \) to the agent position at time $k$, \( \y^k \). The unit vector \( \hat{n}_{o,i}^k \) is normal to the obstacle boundary at this intersection point. The variable \( \theta^k_{o,i} \) represents the angle between the horizontal axis and the line extending from \( p_{o,i} \) toward \( \y^k \). Since these variables depend on \( \y^k \), the time index \( k \) is included in their notation. The following results provide a detailed explanation of how these variables are determined, with some variables illustrated in Fig. \ref{fig: Obstacles}.

\textbf{Circular obstacle: } For a circular obstacle, $r^k_{o,i}$ is simply the radius of the circle and remains constant regardless of $\y^k$.
The unit normal vector $\hat{n}_{o,i}^k$ is given by $\hat{n}_{o,i}^k = (\y^k - p_{o,i}) / \norm{\y^k - p_{o,i}}$.

\textbf{Rectangular obstacle: } For a rectangular obstacle, $\bar{\theta}_{o,i}$ is given by $\atantwo (l_{W,i},l_{L,i})$. Then $r^k_{o,i}$ is obtained as follows:
{\small
\begin{flalign}
    &r^k_{o,i} = \left\{\ \begin{aligned}
    &\frac{l_{W,i}}{2\abs{\cos(\pi/2-\theta^k_{o,i})}}, &&\text{if } \bar{\theta}_{o,i}<\abs{\theta^k_{o,i}}<\pi-\bar{\theta}_{o,i}
    \\ &\frac{l_{L,i}}{2\abs{\cos(\theta^k_{o,i})}}, &&\text{otherwise}.
    \end{aligned}\right.&&\nonumber
\end{flalign}}
The unit normal vector $\hat{n}_{o,i}^k$ is defined as follows: $[1, 0]^\top$ if $-\bar{\theta}_{o,i} \leq \theta^k_{o,i} < \bar{\theta}_{o,i}$, $[0, 1]^\top$ if $\bar{\theta}_{o,i} \leq \theta^k_{o,i} < \pi - \bar{\theta}_{o,i}$, $[0,-1]^\top$ if $\bar{\theta}_{o,i} - \pi \leq \theta^k_{o,i} < -\bar{\theta}_{o,i}$, and $[-1, 0]^\top$ otherwise.

\subsection{Control Barrier Functions (CBFs)}
CBFs are a mathematical framework developed for enforcing safety constraints in control systems. They ensure that the trajectory of the dynamic system remains within the predefined safe set while allowing flexibility in control design.
Consider the following nonlinear control-affine system:
\begin{align}
    \x^{k+1} = f(\x^k)+g(\x^k)u^k,\ \y^k = H_\y\x^k,\ v^k = H_v\x^k, \label{eqn: nonlinear control affine system}
\end{align}where $\x\in D\subset\R^{n}$ and $u\in U\subset\mathbb{R}^m$ represent the state and the input vector, where $ D$ and $U$ denote the sets of admissible states and inputs, respectively. The functions $f:\R^n\rightarrow\R^n$ and $g:\R^n\rightarrow\R^{n\times m}$ denote system functions, and $\y \in \mathbb{R}^{p}$ and $v \in \mathbb{R}^{p}$ represent the output vectors corresponding to position and velocity, respectively, in a $p$-dimensional Cartesian space.

Let the set $\C\subseteq D$ denote the \textit{safe set} of states, i.e., free from collisions. The continuously differentiable function $h:D\rightarrow\mathbb{R}$ can be constructed such that safe set $\C$ is its superlevel set, defined as follows: 
\begin{definition}{(Superlevel set \citep{ames2019control})}
    The set $\C$ is defined as a superlevel set of a continuously differentiable function $h: D\subset \mathbb{R}^n \rightarrow \mathbb{R}$:
    \begin{align}
        \begin{aligned}
            \C = \{\x \in D \subset \mathbb{R}^n: h(\x) \geq 0\}. 
        \end{aligned}\label{eqn: CBF_safeset}
    \end{align}
    Furthermore, the boundary and interior of $\C$ are defined as
    $
            \partial\C = \{\x \in D \subset \mathbb{R}^n: h(\x) = 0\}$ and $
            \text{Int}\C = \{\x \in D \subset \mathbb{R}^n: h(\x) > 0\},
$ respectively.
\end{definition}

The function $h$ is referred to as a \textit{control barrier function} if it satisfies the condition specified in the definition below.
\begin{definition}{(Control barrier function \citep{ames2019control})}
Consider a continuously differentiable function $h: D\rightarrow\mathbb{R}$ where $\C\subset D\subseteq\mathbb{R}^n$ is its superlevel set. The function $h(\x)$ is a control barrier function (CBF) if there exists an extended class-$\mathcal{K}_\infty$ function $\alpha$ that satisfies for the system \eqref{eqn: nonlinear control affine system}:
\begin{align}
    \exists\ u\in U \ \  \text{s.t.}\  \dot{h}(\x,u) \geq -\alpha(h(\x))\  \Leftrightarrow \ \C\ \text{is invariant}\label{eqn: ProbStat_CBF_Cond}
\end{align}
\end{definition}
In the discrete-time system, \eqref{eqn: ProbStat_CBF_Cond} is represented by 
\begin{align}
  h(\x^{k+1})-h(\x^k) \geq -\alpha(h(\x^k)). \label{eqn: CBF_safecond}   
\end{align}
From the control barrier function $h$, consider the set of inputs that satisfies \eqref{eqn: ProbStat_CBF_Cond}, defined as
\begin{align}
   K_{\text{cbf}}(\x) = \{u\in U: \dot{h}(\x,u)+\alpha(h(\x))\geq0\}.\label{eqn: ProbStat_Kcbf}
\end{align}
Based on the earlier definitions, the following theorem guarantees the system's safety.
\begin{theorem}{(Condition for the safety of the dynamical system \citep{ames2019control})}\label{thm: ProbStat_safety condition}
    Let $\C\subset D$ be a compact superlevel set of a continuously differentiable control barrier function $h:D\rightarrow\mathbb{R}$ on $D$. If $\frac{\partial h}{\partial \x}(\x) \neq 0$ for all $\x\in\partial \C$, then any Lipschitz continuous control input $u = K_{\text{cbf}}(\x)\in U$ for the control-affine system \eqref{eqn: nonlinear control affine system} renders the set $\C$ safe. Additionally, the set $\C$ is asymptotically stable in $D$. 
\end{theorem}

Theorem \ref{thm: ProbStat_safety condition} guarantees that if the initial state $\x_0$ belongs to the safe set $\C$, then the state $\x(t)$ remains in $\C$ for all $t\in[0,\infty)$ under any control input satisfying \eqref{eqn: ProbStat_Kcbf}, ensuring the safety of the dynamical system.

\subsubsection{Design of Control Barrier Functions} 

Based on the obstacle models in Section \ref{sec: Obstacle Models}, two CBFs are proposed for each obstacle $i$ as follows:
{\small
\begin{flalign}
    &h_{i,1}(\x^{k}) = \norm{H_\y\x^{k}-p_{o,i}}^2-(r_{o,i}^{k-1})^2,&& \label{eqn: CBF_h1}
\\&
    h_{i,2}(\x^{k}) = \norm{H_\y\x^{k}-p_{o,i}}^2-(r_{o,i}^{k-1})^2+K_v\dotprod{\operatorname{\hat{n}}^{k-1}_{o,i}}{H_v \x^{k-1}}\label{eqn: CBF_h2}&&
    \\&/\left\{\frac{\norm{{H_\y\x^{k-1}-p_{o,i}}}-r_{o,i}^{k-1}}{\norm{{H_\y\x^{k-1}-p_{o,i}}}}\dotprod{\operatorname{\hat{n}}^{k-1}_{o,i}}{{H_\y\x^{k-1}-p_{o,i}}}\right\}.\nonumber&&
\end{flalign}
}
The first CBF, defined in equation~\eqref{eqn: CBF_h1}, ensures that the agent remains safe by avoiding areas inside the obstacles, designated as unsafe sets. The second CBF, given by equation~\eqref{eqn: CBF_h2}, enhances the robustness of the avoidance strategy by incorporating factors not explicitly captured in the agent's dynamics, such as acceleration constraints. For instance, while assuming unbounded acceleration in theoretical models allows instant direction changes, this is not physically realistic. In practice, acceleration limits may prevent the agent from decelerating in time to avoid a collision. To address this, the second CBF adjusts the safe set by shrinking the permissible velocity in the direction of the obstacle as the agent approaches, ensuring sufficient time for deceleration and collision avoidance.

\textbf{Example:} Consider an obstacle located in a 1D domain, as depicted in Fig.~\ref{fig: Exam_obstacle}. The agent moves along this axis. The position and velocity of the agent are represented by $\y^k$ and $v^k$, respectively. The parameters shown in Fig.~\ref{fig: Obstacles} are set as follows: $p_o = 0$, $r_o^k = 2$, and $\hat{n}_o^k = \text{sign}(\y^k)$. Additionally, the parameter $K_v$ is set to 5. The resulting safe sets, defined by \eqref{eqn: CBF_h1} and \eqref{eqn: CBF_h2}, are illustrated in Fig.~\ref{fig: Example_Safeset}(a) and (b), where the green-shaded areas represent the safe sets defined by \eqref{eqn: CBF_safeset}.

\begin{figure}[h]
    \centering
    \includegraphics[width=0.5\linewidth]{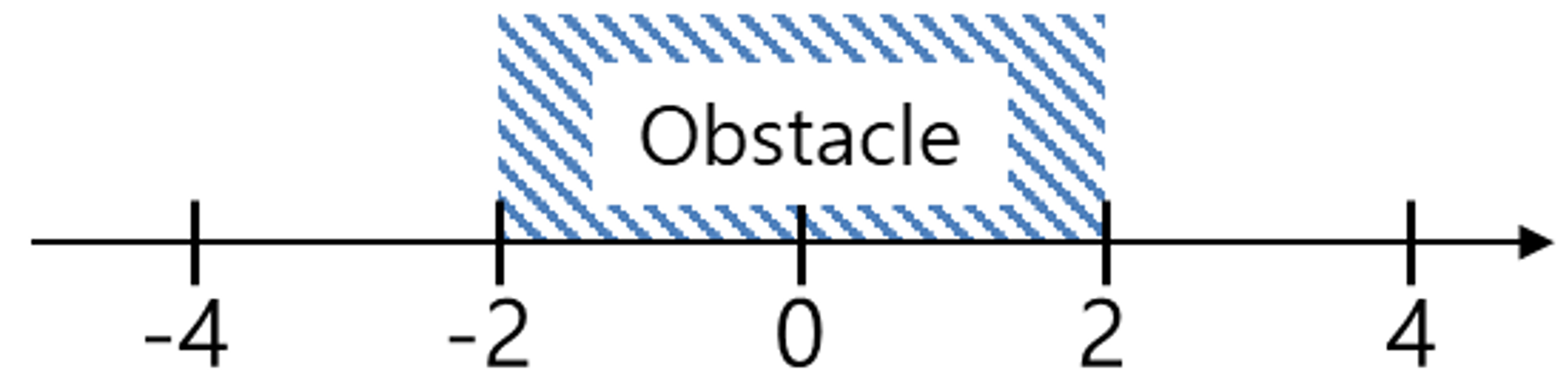}
    \caption{Example: an obstacle in a one-dimensional domain}
    \label{fig: Exam_obstacle}
\end{figure}

\begin{figure}[h]    
    \centering    
    \subfloat[Safe set of $h_{i,1}$]{    
        \includegraphics[width=0.310\linewidth]{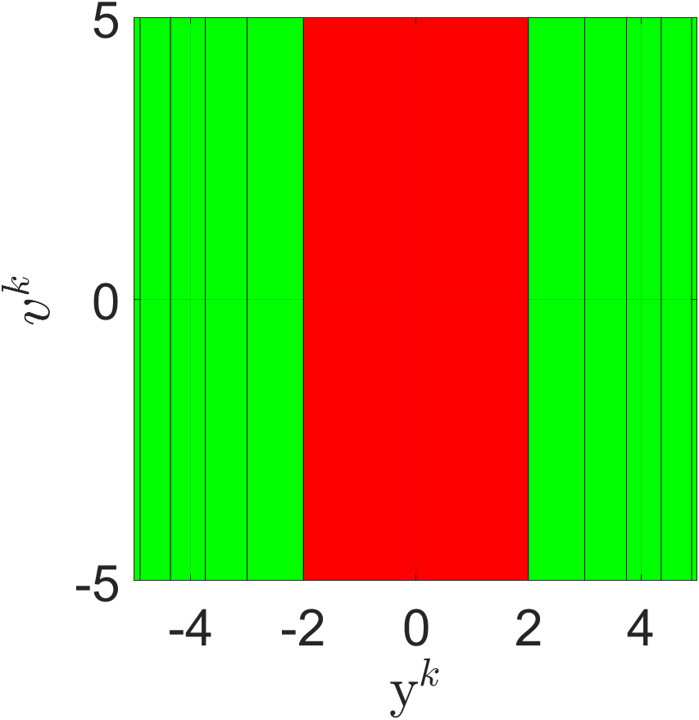}}    
    \   
    \subfloat[Safe set of $h_{i,2}$]{    
        \includegraphics[width=0.310\linewidth]{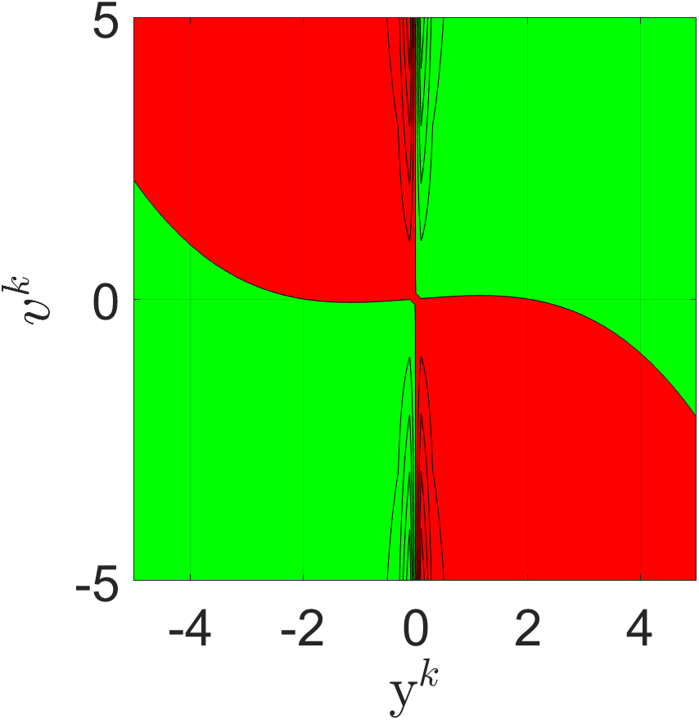}}    
    \   
    \subfloat[Overall safe set]{    
        \includegraphics[width=0.310\linewidth]{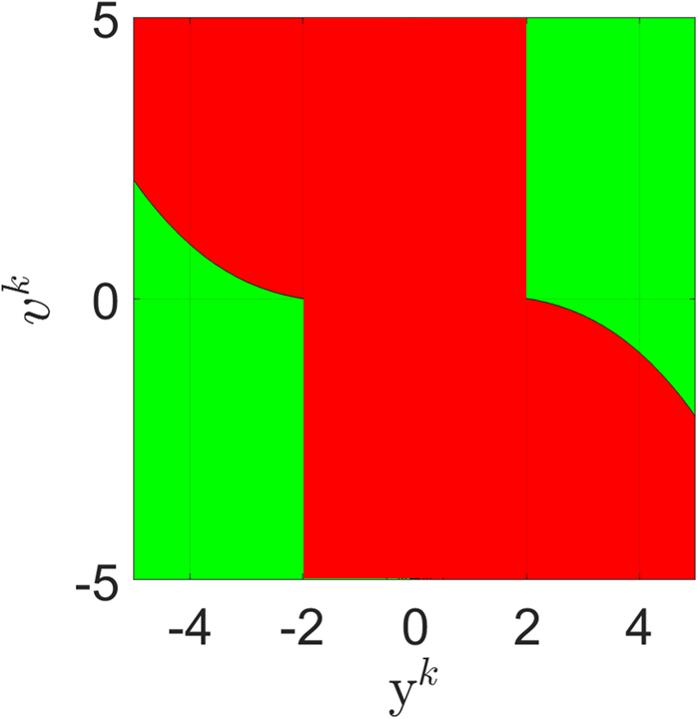}} 
    \caption{Safe sets defined by $h_{i,1}$ in \eqref{eqn: CBF_h1}, $h_{i,2}$ in \eqref{eqn: CBF_h2} and the resulting overall safe set.}    
    \label{fig: Example_Safeset}
\end{figure}

In Fig.~\ref{fig: Example_Safeset}(a), the safe set of $v^k$ remains constant regardless of the agent's position, $\y^k$. In contrast, Fig.~\ref{fig: Example_Safeset}(b) illustrates that the safe set of $v^k$ in the direction toward the obstacle shrinks as the agent approaches it, while the safe set in the opposite direction remains unbounded. This shrinkage forces the agent to reduce its velocity when approaching the obstacle, enhancing its robustness against acceleration constraints unknown to the controller and improving its safety during navigation. 

By utilizing both CBFs, the resulting overall safe set is defined as the intersection of the safe sets shown in Fig.~\ref{fig: Example_Safeset}(c), ensuring that both collision avoidance and robust navigation are achieved.

Before presenting the safety guarantee using CBFs, it is essential to first examine the concept of output relative degree in discrete-time nonlinear systems. This notion specifies the number of discrete-time differences required for the system output to become explicitly dependent on the control input. A clear understanding of the relative degree is critical for the design of safety-critical controllers, such as CBFs, as it determines when and how control inputs can affect constraint enforcement.

\begin{definition}{(Output Relative Degree $P$ in Discrete-Time Nonlinear Systems)}\label{def: output relativedegree}
Consider the discrete-time nonlinear control-affine system \eqref{eqn: nonlinear control affine system}.
The \emph{output relative degree} $P \in \mathbb{N}$ is defined as the smallest positive integer such that the input $u^k$ explicitly affects the future output $\y^{k+P}$. Specifically, it is the smallest $P$ such that
\[
\frac{\partial \y^{k+P}}{\partial u^k} = H_\y \cdot \frac{\partial \x^{k+P}}{\partial u^k} \neq 0,
\]
and
\[
\frac{\partial \y^{k+i}}{\partial u^k} = H_\y \cdot \frac{\partial \x^{k+i}}{\partial u^k} = 0, \quad \forall i < P.
\]

This condition reflects the number of discrete-time steps required for the control input \( u^k \) to influence the output \( \y^{k+P} \) through the system dynamics.
\end{definition}

\begin{theorem}{(Safety Guarantee via CBFs)}\label{thm: discrete-time Pos_Invar}
   Consider the nonlinear control-affine system in \eqref{eqn: nonlinear control affine system} with output relative degree \( P \). Let \( \mathcal{C} \) be the safe set associated with a CBF, \( h \), defined by either \eqref{eqn: CBF_h1} or \eqref{eqn: CBF_h2}. Given any initial state \( \x^0 \in \mathcal{C} \), if the invariance condition \eqref{eqn: CBF_safecond} holds for \( k = 0, 1, \dots, P-2 \), then applying any control input \( u^{k-P+1} \in U \) that satisfies \eqref{eqn: CBF_safecond} for all \( k \in \mathbb{I}_{P-1:\infty} \) ensures that the state of the system remains in \( \mathcal{C} \).
\end{theorem}

\begin{proof}
    By the sufficient condition for \( k = 0, 1, \dots, P-2 \) in Theorem \ref{thm: discrete-time Pos_Invar}, the states \( \x^1, \x^2, \dots, \x^{P-1} \) remain in \( \mathcal{C} \) as guaranteed by condition \eqref{eqn: CBF_safecond}. 

For all \( k \in \mathbb{I}_{P-1:\infty} \), the invariance condition \eqref{eqn: CBF_safecond} can be expressed as a function of \( \x^{k-P+1} \) and \( u^{k-P+1} \). The state-related terms in \eqref{eqn: CBF_h1} and \eqref{eqn: CBF_h2} can be rewritten using the control-affine system dynamics. These terms simplify to the following expressions:
\[
H_\y(\x^{k+1}) = H_\y f(\x^{k}) + H_\y g(\x^{k}) u^{k} \\
\]
and similarly for other outputs. Applying these results and using the relative degree property (as defined in Definition \ref{def: output relativedegree}), we can verify that the control input \( u^{k-P+1} \), when it satisfies the invariance condition \eqref{eqn: CBF_safecond}, ensures that the state \( \x^k \) remains within the safe set \( \mathcal{C} \) for all \( k \in \mathbb{I}_{P-1:\infty} \).
\end{proof}

By choosing the function $\alpha(h(\x^k)) = h(\x^k)$, the condition \eqref{eqn: CBF_safecond} is rewritten by $h(\x^{k+1}) \geq 0$. Then the admissible set of control input obtained based on Theorem \ref{thm: discrete-time Pos_Invar} is defined as follows:
{
\begin{align}
&K_{\text{cbf}}(\x^{k+P-1})\label{eqn: discrete_K_cbf} = \\&\{u^{k}\in U: 
h_{i,1}(\x^{k+P})\geq0,h_{i,2}(\x^{k+P})\geq0, \forall i\},\ \ \forall k\in\mathbb{N}\nonumber
\end{align}
}

Based on the control input obtained from the D$^2$C scheme for achieving area coverage, safe navigation can be ensured by minimally modifying this input so that it belongs to the set defined in \eqref{eqn: discrete_K_cbf}. This approach allows agents to avoid obstacles without significantly degrading area coverage performance. To do this, the modified control input $\tilde{u}^{k}$ is computed by solving the optimization as follows:
\begin{align}
    \begin{aligned}
    \tilde{u}^{k}=\argmin\nolimits_{u\in K_{\text{cbf}}(\x^{k+P-1})}  \norm{u-u^{k}}^2
    \end{aligned}\label{eqn: CBF_Optimization}
\end{align}
where $u^k$ is the control input obtained by the D$^2$C scheme. This optimization problem can be solved using a quadratic programming solver.

\section{Simulation}
To evaluate the performance of the proposed CBF-integrated D$^2$C scheme, two simulations were conducted: 1) a comparison of the CBF approach with and without the velocity-constrained CBF $h_{i,2}$, and 2) a comparison between D$^2$C+APF and D$^2$C+CBF in terms of area coverage and safety. The simulations utilize a discrete-time linear quadrotor model \citep{sabatino2015quadrotor}. The state vector is defined as $\x^{k} := [p_{x}^{k} \ dp_{x}^{k} \ \theta^{k} \ d\theta^{k} \ p_{y}^{k} \ dp_{y}^{k} \ \phi^{k} \ d\phi^{k}]^{\top}$, with the input vector given by $u^k = [\tau^k_{x} \ \tau^k_{y}]^{\top}$. The output vectors are defined by $\y^k = [p^k_{x} \ p^k_{y}]$ and $v^k = [dp_{x}^{k} \ dp_{x}^{k}]$. Here, $p_{x}$ and $p_{y}$ represent the quadrotor’s position in the $x$ and $y$ axes, $\phi$ and $\theta$ are the roll and pitch angles, and $\tau_{x}$ and $\tau_{y}$ denote the torques applied along the $x$- and $y$-axes, respectively. The notation $d(\cdot)$ represents the change in the variable between consecutive time intervals.

The control inputs are constrained as $\abs{\tau_{x}^{k}}, \abs{\tau_{y}^{k}}<10 \text{ Nm}$. The state variables are constrained as follows: $\abs{dp_{x}^{k}}, \ \abs{dp_{y}^{k}} < 1.75 \ \text{m/s}$, $\abs{\phi^{k}}, \ \abs{\theta^{k}} < 1.5^\circ$, and $\abs{d\phi^{k}}, \ \abs{d\theta^{k}} < 15^\circ/\text{sec}$. The communication range of each agent is set to 100 meters. The optimization problem defined in \eqref{eqn: CBF_Optimization}, subject to the constraint \eqref{eqn: discrete_K_cbf}, is solved using the \texttt{fmincon} function in \textsc{MATLAB}.
\subsection{Impact of Velocity-Constrained CBF on Safe Naviga-\\tion}

The CBF \( h_{i,1} \) enforces positional safety by defining a hard boundary around obstacle \( i \), while \( h_{i,2} \) further enhances safety by regulating the agent’s velocity as it approaches the obstacle. To demonstrate the effectiveness and robustness of combining both CBFs, we performed simulations using a linearized quadrotor model under two scenarios: (i) employing only \( h_{i,1} \), and (ii) employing both \( h_{i,1} \) and \( h_{i,2} \).

The simulation results are presented in Fig.~\ref{fig: Sim_Robustness}, where the black solid lines depict the agent’s trajectories, and red rectangles indicate the obstacles. The gray dots represent the reference distribution in its sample-based form, while the initial and final positions are marked by blue crosses and yellow circles, respectively.

\begin{figure}[h]
    \centering
    \includegraphics[width=0.8\linewidth]{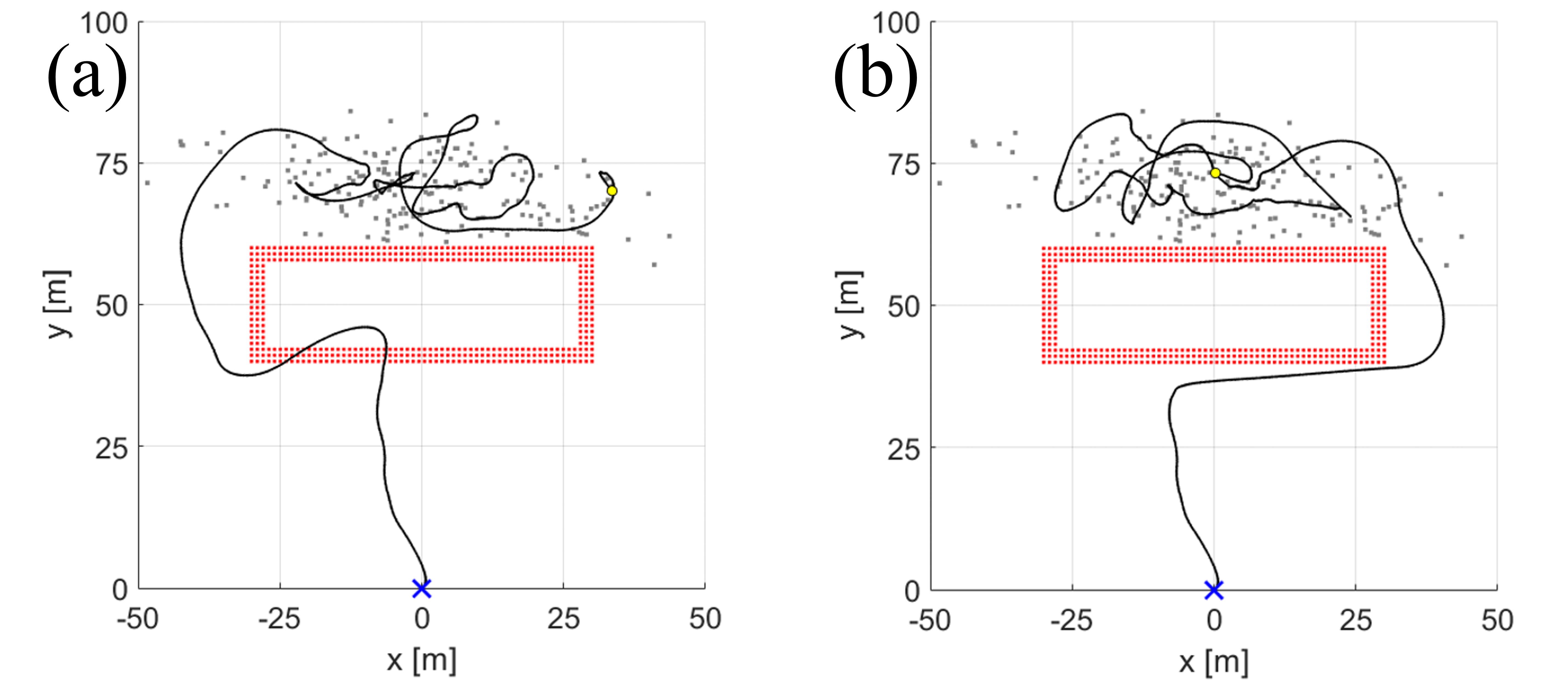}
    \caption{Impact of incorporating $h_{i,2}$ on collision avoidance under state constraints. (a) Only \( h_{i,1} \) is employed. (b) Both \( h_{i,1} \) and \( h_{i,2} \) are employed.}
    \label{fig: Sim_Robustness}
\end{figure}


In Fig.~\ref{fig: Sim_Robustness}(a), the agent momentarily violates the obstacle boundary, as it fails to reduce velocity in time. This failure arises from the controller's inability to account for implicit state constraints such as maximum deceleration or turning limitations, which stem from bounds on angular velocity and heading rate. Without anticipating these constraints, the agent cannot adapt its trajectory quickly enough to avoid intrusion.

In contrast, Fig.~\ref{fig: Sim_Robustness}(b) shows that when both CBFs are used, the agent smoothly decelerates as it nears the obstacle and maintains a safe distance throughout its trajectory. The inclusion of \( h_{i,2} \) enables proactive velocity modulation, accounting for the agent’s physical limitations and ensuring safer, more robust navigation.

\subsection{Safe Area Coverage: D$^2$C with APF- and CBF-Based Collision Avoidance}

To assess the effectiveness of safe area coverage using the D$^2$C framework in environments with obstacles, simulations were conducted in a domain featuring static obstacles and SPs. In Fig.~\ref{fig: Coverage_result}, gray dots represent the SPs, blue crosses indicate the initial positions of the agents, and red shapes denote the obstacles.

The study compares two collision avoidance schemes: D$^2$C integrated with APF and D$^2$C augmented with CBFs. For further details on the APF-based collision avoidance method, refer to~\cite{seo2023density}. The simulation involves three agents, with inter-agent collision avoidance modeled by treating each agent as a dynamic circular obstacle centered at 
$p_o=\y^k$, 
with radius $r^k_o= 5$ m representing the safety margin to prevent actual collisions. A collision is considered to occur if any two agents come within 5 m of each other. 

\begin{figure}[h]
    \centering
    \includegraphics[width=0.8\linewidth]{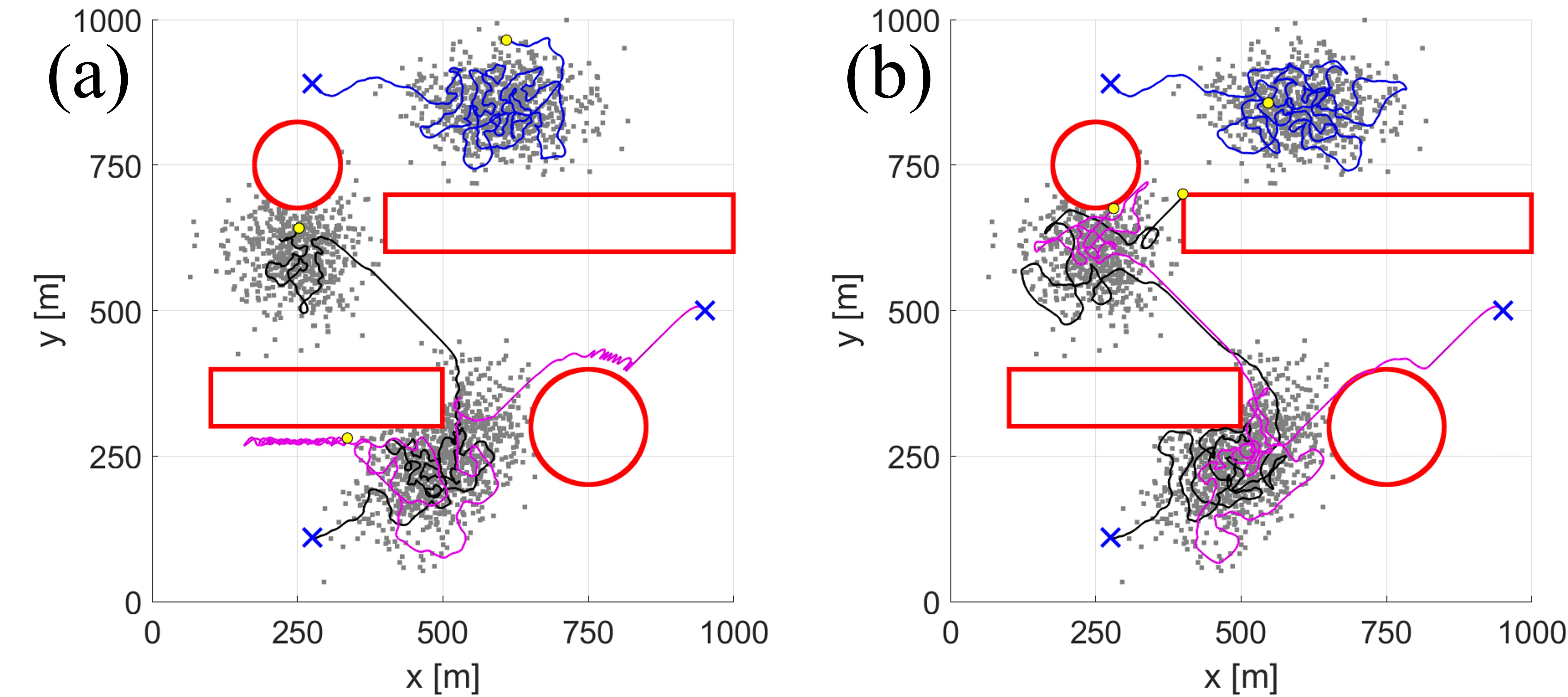}
    \caption{Safe area coverage using D$^2$C : (a) with APF; (b) with CBFs. Agent trajectories are shown as blue, black, and magenta lines.}
    \label{fig: Coverage_result}
\end{figure}


As shown in Fig.~\ref{fig: Coverage_result}(a), D$^2$C+APF exhibits oscillatory behavior near obstacle boundaries. This arises from conflicts between the repulsive forces of the potential field and the agents' attraction to SPs, which leads to overly cautious clearance around obstacles and difficulty in navigating narrow passages. Since obstacle boundaries are inferred from the shape of the potential field rather than explicitly defined, this behavior limits efficient movement. The minimum observed inter-agent distance during the operation was $12.8$ m.

In contrast, Fig.~\ref{fig: Coverage_result}(b) illustrates that D$^2$C+CBFs enables smooth, collision-free navigation. By explicitly defining obstacles as unsafe sets through $h_{i,1}$, and shrinking the safe velocity set near them using $h_{i,2}$, the velocities of the agents are appropriately adjusted near obstacles, allowing them to maneuver more closely and efficiently while maintaining safety. This results in both more precise coverage and better safety enforcement. The minimum inter-agent distance observed throughout the simulation was $14.7$ m, ensuring no collisions occurred.

To quantify area coverage performance, the Wasserstein distances between the agent trajectories and the reference distribution were computed for both simulations. The distance for D$^2$C+APF was 130.67, while for D$^2$C+CBF it was 51.60. Since the Wasserstein distance measures the dissimilarity between distributions, a lower value indicates better coverage. Thus, D$^2$C+CBF demonstrated more efficient area coverage than D$^2$C+APF.

\section{Conclusion}

This paper presents a novel approach for safe and efficient multi-agent non-uniform area coverage by integrating D$^2$C with CBFs. A key contribution is the generalized CBF formulation, enabling safe navigation around rectangular obstacles, with the unit normal vector analytically derived to ensure accurate safety enforcement. Additionally, a velocity-dependent term improves collision avoidance by adapting to dynamic environments. Simulation results demonstrate that the D$^2$C$+$CBF approach achieves better area coverage, as a result of its smoother trajectories near obstacles, compared to the D$^2$C$+$APF baseline in our simulation setup, while ensuring safe navigation. This work offers a practical solution for multi-agent systems in complex environments, with potential for further research on multi-agent coordination and more complex scenarios.

\bibliography{ifacconf.bib}             
                                                   







\end{document}